\documentclass[a4paper,twoside]{article}

\usepackage{epsfig}
\usepackage{subcaption}
\usepackage{calc}
\usepackage{amssymb}
\usepackage{amstext}
\usepackage{amsmath}
\usepackage{amsthm}
\usepackage{multicol}
\usepackage{pslatex}
\usepackage{apalike}
\usepackage[linesnumbered,ruled,vlined]{algorithm2e} % added
\DeclareMathOperator*{\argmax}{arg\,max} % added
\usepackage[ruled,vlined,linesnumbered]{algorithm2e} % added
\usepackage[bottom]{footmisc}
\newtheorem{theorem}{Theorem} % added
\newtheorem{definition}{Definition} % added
\usepackage{SCITEPRESS}     % Please add other packages that you may need BEFORE the SCITEPRESS.sty package.

\begin{document}

\title{Fairness Driven Multi-Agent Path Finding Problem}

\author{\authorname{Aditi Anand \sup{1}\orcidAuthor{0009-0001-7280-7335}, Dildar Ali \sup{1}\orcidAuthor{0000-0002-3427-1904} and Suman Banerjee \sup{1}\orcidAuthor{0000-0003-1761-5944}}
\affiliation{\sup{1} Department of Computer Science and Engineering, Indian Institute of Technology Jammu, Jammu, Jammu and Kashmir, 181221, India}
\email{\{2022ucs0076, 2021rcs2009, suman.banerjee\}@iitjammu.ac.in}
}


\abstract{The Multi-Agent Path Finding (MAPF) problem aims at finding non-conflicting paths for multiple agents from their respective sources to destinations. This problem arises in multiple real-life situations, including robot motion planning and airspace assignment for unmanned aerial vehicle movement. The problem is computationally expensive, and adding to it, the agents are rational and can misreport their private information. In this paper, we study both variants of the problem under the realm of fairness. For the non-rational agents, we propose a heuristic solution for this problem. Considering the agents are rational, we develop a mechanism and demonstrate that it is dominant strategy incentive compatible and individually rational. We employ various solution methodologies to highlight the effectiveness and efficiency of the proposed solution approaches.}

\onecolumn \maketitle \normalsize \setcounter{footnote}{0} \vfill

\section{Introduction} \label{Sec:Intro}
The Multi-Agent Path Finding (MAPF) problem involves finding non-conflicting paths for each agent, which is particularly relevant to specific applications such as robotics and inventory management \cite{liu2024multi}, video games, vehicle/UAV routing \cite{ho2019multi}, and the pickup and delivery problem \cite{ma2017lifelong}, among others. In this problem, we are given a set of agents and an underlying graph where the start and destination vertices of every agent are marked. The problem is to find a path for every agent such that some conflicting constraints are satisfied. Such constraints will vary from application to application. In the case of warehouse robotics, the constraint may be that there should not be any obstacles in the proposed path of a robot. In the case of route planning for a fleet of UAVs, the constraint may be that two UAVs cannot move to each other at the same timestamp. Due to the large number of potential applications MAPF Problem has been studied extensively in different variants \cite{stern2019multi} for a large number of agents, in online set up where an agent can join at a later stage as well \cite{vsvancara2019online}, where the goal is to find out non-optimal path in polynomial time, considering multiple objectives simultaneously, considering the presence of uncertainty \cite{shahar2021safe}, considering delay probabilities \cite{ma2017multi}, combining learning-based techniques \cite{sartoretti2019primal}, and many more. A large number of solutions based on heuristic search \cite{okumura2023lacam} are used to solve this problem.  
% \par Though there exist several studies on the MAPF Problem, to the best of our knowledge, the notion of fairness has not been considered. 
% In practice, there are many situations where fairness must be considered to maintain the stability of the entire system and bring out a socially relevant solution. Here, we present two situations where ensuring fairness is crucial. 
% \begin{itemize}
% \item Consider a fleet of UAVs that has been deployed for surveillance of a whole security zone. The duration of a UAV's movement is directly proportional to the battery power available to it. Hence, if there is a significant variance in the path length, it may be the case that the battery power of several UAVs from the fleet will be underutilized. For some UAVs, as their battery power decreases, they become unusable for surveillance, which may result in some zones being out of coverage. If we ensure fairness in this context, such that the distance traveled by any pair of UAVs does not differ too much, then it will ensure a stable system. 
% \item Consider a restaurant that has employed a number of delivery boys (referred to as agents subsequently). In this context, the welfare of an agent is defined as the difference between the currency received for completing the delivery and the cost involved in completing these deliveries.
% \end{itemize}
% However, to the best of our knowledge, the MAPF Problem has not been studied under the realm of fairness. 

Although several studies exist on the MAPF problem, to the best of our knowledge, the notion of fairness has not been considered. In this paper, we bridge this gap by studying the problem under different constraints. In particular, we make the following contributions in this paper:
\begin{itemize}
\item We study the MAPF problem under three fairness constraints, namely, $\epsilon$-envy freeness, max-min fairness, and proportional fairness.
\item  We have proposed solutions for non-rational agents and developed a dominant strategy incentive compatible and individually rational mechanism for rational agents.
\item We have conducted a large number of experiments using MAPF Benchmark datasets to demonstrate the effectiveness and efficiency of the proposed solution approaches.
\end{itemize}

The rest of this paper is organized as follows. Section~\ref{Sec:Background} presents the problem background. Section~\ref{Sec:Proposed_solution} describes the proposed solution methods. Section~\ref{Sec:mechanism_design} discusses the mechanism design for the MAPF problem. Section~\ref{Sec:Exp_Evualuation} reports the experimental evaluation and results. Finally, Section~\ref{sec:conclusion} concludes the paper and outlines future research directions.

\section{Background and Problem Set Up} \label{Sec:Background}
% In this section, we describe the background of the problem and define it formally.
\paragraph{Basic Notations and Terminologies.} We denote a graph by $\mathcal{G}(\mathcal{V},\mathcal{E})$ where $\mathcal{V}$ and $\mathcal{E}$ denotes the vertex and edge set of $\mathcal{G}$, respectively. We denote the number of vertices and edges by $n$ and $m$, respectively. All the graphs considered in this paper are simple, undirected, and unweighted. For every vertex $v \in V$, the neighborhood of $v$ is denoted by $\mathcal{N}(v)$ and defined as the set of vertices with which $v$ has an edge in $G$. For any positive integer $\ell$, $[\ell]$ denotes the set $\{1,2, \ldots, \ell\}$.

% The degree of $v$ is denoted by $deg(v)$ and defined as the cardinality of $\mathcal{N}(v)$, i.e., $deg(v)=|\mathcal{N}(v)|$. 

% The symbols and notations used in this paper are listed in Table \ref{Table1:Notations}.
\paragraph{Agents and Problem Instance.} We consider there are $\ell$ agents $\mathcal{A}=\{a_1, a_2, \ldots, a_{\ell}\}$ and each agent has its start and destination vertex. For any $a_i \in \mathcal{A}$, its start and destination vertex is denoted by $v_{s_i}$ and $v_{d_i}$, respectively. In our study, we consider that time is discretized into equal durations, called `time stamps'. At any particular time stamp, an agent can either stay at the same vertex or move to one of the neighboring vertices. For agent $a_i$, the cost for staying at the same vertex or moving to any neighboring vertex is $c^{step}_i$. Also, every agent $a_i$ has a value $u_i$ that the agent $a_i$ will get after reaching the destination vertex. All these information are accumulated as $[\tau_{i}=(a_i, v_{s_i}, v_{d_i}, u_i, c^{step}_i)]^{\ell}_{i=1}$ and is called agent `type'.

\paragraph{Welfare of an Agent.} From one time stamp to the subsequent time stamp, an agent can stay at the same vertex or move to a neighboring vertex. We assume that starting at time $t=0$, all agents are at some `garbage' vertex. Each agent is requesting a path from its initial to the destination vertex. Once an agent reaches its destination vertex, it disappears. Path of agent $a_i$  ($\pi_{i}$) is defined as a sequence of vertices with the respective timestamps. The set of vertices and edges in path $\pi_{i}$ are denoted by $\mathcal{V}(\pi_{i})$ and $\mathcal{E}(\pi_{i})$, respectively. The total cost of movement is $\mathcal{C}_{i}(\pi_i)=|\pi_i| . c_i^{step}$ and \emph{welfare} of the agent $a_i$ is defined as $W_{i}(\pi_{i})=u_i - \mathcal{C}_{i}(\pi_i)$, where $|\pi_{i}|$ denotes the length of the path. Next, we define social welfare.

% Assume that $\pi_{i}=<(v^{1}_{i}, t_i), (v^{2}_{i}, t_i+1), \ldots, (v^{d}_{i}, t^{'}_i)>$, $(t^{'}_{i}-t_i)$ is the time required for $a_i$ to reach the destination vertex. For any two consecutive time stamped vertices $(v^{x}_i,t)$ and $(v^{y}_i,t+1)$, $v^{y}_i=v^{x}_i$ or $v^{y}_i \in \mathcal{N}(v^{x}_i)$. 

\begin{definition}[Social Welfare]\label{Def:Social_Welfare}
Social welfare is defined as the sum of the welfare of all agents. 

\begin{equation}\label{Eq:Social_Welfare}
SW(\pi_1, \pi_2, \ldots, \pi_{\ell})=\underset{i \in [\ell] }{\sum} \ W_{i}(\pi_{i})
\end{equation}
\end{definition}
\paragraph{Conflict Constraints.} Motivated by numerous real-world applications, ranging from robot motion planning to airspace assignment for unmanned aerial vehicle (UAV) movement, several conflicting constraints must be maintained. In this study, we are interested in the following two types of constraints:
\begin{itemize}
\item \textbf{Vertex Conflicting Constraint:} A vertex conflict occurs when two agents are at the same vertex at the same time stamp. A solution $\Pi$ is vertex conflict free if for any pair of paths $\pi_i, \pi_j \in \Pi$, $\mathcal{V}(\pi_i) \cap \mathcal{V}(\pi_j)=\emptyset$.
\item \textbf{Edge Conflicting Constraint:} Edge conflict occurs when two agents want to traverse the same edge between the same consecutive time stamps. To ensure that the solution $\Pi$ is edge conflict-free, we need to verify for any $v_x \in \mathcal{V}$ and $t \in [T-1]$, if $(v_x,t), (v_y,t+1) \in \mathcal{V}(\pi_i)$ then $(v_y,t), (v_x,t+1) \notin \mathcal{V}(\pi_j)$ for any $i, j \in [l]$. 
\end{itemize}
Next, we define the feasible and optimal path assignments in Definitions \ref{Def:Feasible_path} and \ref{Def:Optimal_path}.
\begin{definition}[Feasible Path Assignment]\label{Def:Feasible_path}
A path assignment $\Pi=(\pi_1, \pi_2, \ldots, \pi_{\ell})$ is said to be feasible if it obeys the vertex and edge conflicting constraints. 
\end{definition}

\begin{definition}[Optimal Path Assignment]\label{Def:Optimal_path}
An assignment is said to be optimal if the aggregated welfare is maximized, as stated in Equation \ref{Eq:optimal}.
\begin{equation} \label{Eq:optimal}
(\pi_{1}^{*}, \pi_{2}^{*}, \ldots, \pi_{\ell}^{*}) \longleftarrow \underset{\Pi = (\pi_1, \pi_2, \ldots, \pi_{\ell})}{argmax} \ \underset{i \in [\ell]}{\sum} \ W_{i}(\pi_{i})
\end{equation}
\end{definition}

The optimal aggregated welfare is denoted by $d^{*}$, i.e., $d^{*}= \underset{i \in [\ell]}{\sum} W_{i}(\pi^{*}_{i})$. In existing studies, Friedrich et al. \cite{friedrich2024scalable} investigated a variant of the MAPF problem and proposed several mechanisms to address it. The solution methodologies proposed by them may lead to situations where there are huge discrepancies in individual welfare among the agents. This is one of the key limitations of their study, which we have addressed in this paper. 

\subsection{Fairness in MAPF Problem} As described in Section \ref{Sec:Intro}, ensuring fairness is very important in the MAPF Problem. Now, we describe the fairness notions considered in this paper.

% \paragraph{\textbf{$\epsilon$-Envy Freeness.}} 
% The $\epsilon$-envy-freeness ensures that no agent envies the allocation of another. The notion of $\epsilon$-envy-freeness is stated in Definition \ref{Def:Envyfreenes}.
\begin{definition}[$\epsilon$-Envy Freeness]\label{Def:Envyfreenes}
 A solution $\Pi = (\pi_1, \pi_2, \ldots, \pi_\ell)$ is $\epsilon$-envy free if no agent significantly envies the allocation of any other agent, up to a small threshold $\epsilon > 0$. Formally, for all pairs of agents $a_i, a_j \in \mathcal{A}$, the absolute difference in their welfare is bounded by $\epsilon$, i.e.,
 \begin{equation}
     |W_i(\pi_i) - W_j(\pi_j)| \leq \epsilon
 \end{equation}

% This condition ensures that the perceived unfairness between any pair of agents is within an acceptable margin $\epsilon$, thus promoting a bounded envy-free allocation of paths.
\end{definition}

% \paragraph{\textbf{Max-min fairness.}}
% Max-min fairness seeks a feasible solution that maximizes the minimum welfare of all agents as stated in Definition \ref{Def:Max_Min}.
\begin{definition}[Max-min fairness]\label{Def:Max_Min}
Let $\mathcal{F}$ be the set of all feasible path assignments $\Pi' = (\pi'_1, \pi'_2, \ldots, \pi'_\ell)$. Then, max-min fairness wants a solution that maximizes the minimum welfare, that is, 
\begin{equation}
max_{x \in \mathcal{F}} min_{i}W_{i}(\pi'_i)
\end{equation}
% Thus, the minimum welfare is maximized first, then the next, and so on.
% A joint assignment satisfies max-min fairness if no agent can get a higher welfare without reducing the welfare of some other agent who already has equal or lower welfare.
\end{definition}

% \paragraph{\textbf{Proportional Fairness.}}
% A solution is proportionally fair if any alternative allocation that improves the welfare of some agents causes a larger proportional loss to others as stated in Definition \ref{Def:PF}.
\begin{definition}[Proportional Fairness]\label{Def:PF}
A solution $\Pi = (\pi_1, \pi_2, \ldots, \pi_\ell)$ is said to be proportionally fair if for any other feasible solution $\Pi' = (\pi_1', \pi_2', \ldots, \pi_\ell')$, the following inequality holds:
\begin{equation}
\sum_{i=1}^{\ell} \frac{W_i(\pi_i') - W_i(\pi_i)}{W_i(\pi_i)} \leq 0 
\end{equation}
% This implies any alternative allocation that increases the welfare of some agents must result in a proportionally larger decrease in the welfare of others.
\end{definition}

\section{Proposed Solution Approaches}\label{Sec:Proposed_solution}
\subsection{Fair Increasing Cost Tree Search}
We propose an Iterative Deepening $A^*$ ($IDA^*$) based Fair Increasing Cost Tree Search (Fair-ICTS) algorithm to solve the MAPF problem. The proposed approach is presented in Algorithm \ref{Algo1:FairICTS}.
% Main Algorithm 
\begin{algorithm}[h!]
\scriptsize
\SetAlgoLined
\SetKwFunction{DepthBoundSearch}{DepthBoundSearch}
\SetKwFunction{IsProportionallyFair}{IsProportionallyFair}
\SetKwFunction{IsMaxMinFair}{IsMaxMinFair}
\caption{Fair Increasing Cost Tree Search}
\label{Algo1:FairICTS}

\KwIn{Graph $\mathcal{G} = (\mathcal{V}, \mathcal{E})$, 
Agents $\{ \tau_i = (a_i, v_{s_i}, v_{d_i}, u_i, c_i^{step}) \}_{i=1}^{\ell}$, 
Fairness tolerance $\epsilon$}
\KwOut{Fair and conflict-free optimal joint plan $\Pi^\ast$}

$\mathcal{R} \gets \emptyset$\;
Compute $\mathbf{s_{0}} = (s_{1}^*, \ldots, s_{\ell}^*)$ using $A^*$ algorithm\;
Root node $N_0 \gets (\textbf{s}_0, g(N_0)=0, h(N_0))$\;

$\text{depthBound} \gets f(N_0) = g(N_0) + h(N_0)$\;

\While{true}{
    $\text{minExceedingCost} \gets \infty$\;
    \textsc{DepthBoundSearch}(
    $N_0,$ 
    $\text{depthBound},$ 
    $\text{minExceedingCost},$ 
    $\mathcal{R},$ 
    $\epsilon$)\;

    \If{$\text{minExceedingCost} = \infty$}{
        \textbf{break}\;
    }
    \Else{
        $\text{depthBound} \gets \text{minExceedingCost}$\;
    }
}

$\mathcal{R} \gets 
    \{\Pi \in \mathcal{R} \mid$
        \IsProportionallyFair($\Pi, \mathcal{R}$)
        $\land$
        \IsMaxMinFair($\Pi, \mathcal{R}$)
    \}\;

\Return $\Pi^\ast \;=\; \arg\max_{\Pi \in \mathcal{R}} \sum_{i=1}^{\ell} W_i(\pi_i)$\;
\end{algorithm}

% \subsubsection{Search Space Representation}
We represent the search space as a joint space of step vectors over all agents, denoted by $\mathbf{s} = (s_1, \dots, s_{\ell})$, where $s_i$ is the number of time steps allocated to agent $a_{i}$. The root node of the search tree, $\mathbf{s_{0}} = (s_{1}^*, \ldots, s_{\ell}^*)$, corresponds to the minimal step vector in which each $s_i^*$ denotes the shortest possible number of steps required for agent $a_{i}$ to reach its goal. These optimal step counts are computed using $A^{*}$ algorithm \cite{astar} for each agent individually, providing a lower bound on the required number of steps. 
The algorithm then explores this search space of step vectors, constructing and evaluating all possible per-agent paths of corresponding lengths for each vector. Each node $N$ in the search tree corresponds to a specific step vector $\mathbf{s} = (s_1, \dots, s_{\ell})$ and encapsulates feasible paths for each agent within the allocated steps. The objective is to maximize the social welfare, which is equivalent to minimizing the total cost across all agents, while ensuring fairness among them. 

\par We employ the $IDA^*$ framework as the underlying search mechanism, wherein each node N is evaluated using the function \(f(N) = g(N) + h(N)\), where \(g(N)\) denotes the cumulative cost associated with the current joint step allocation, and \(h(N)\) is an admissible heuristic denoting additional cost required to obtain a feasible joint plan. This formulation enables a memory-efficient exploration while preserving the optimality guarantees of the $IDA^*$ framework. At each iteration, all nodes whose $f(N)$ does not exceed the current depth bound are explored. The expansion of a node is performed using the procedure \textsc{ExpandNode} (in Algorithm \ref{Algo2:DepthBoundSearch}). A node $N = (s_1, s_2, \dots, s_\ell)$ has $\ell$ children corresponding to the incremental increase in one agent’s path length. The set of child nodes is given by $\{(s_1 + 1, s_2, \dots, s_\ell), (s_1, s_2 + 1, \dots, s_\ell), \dots, (s_1, s_2, \dots, s_\ell + 1)\}$, ensuring that in each expansion step, only one agent’s path length is extended by one step while others remain unchanged. 

\par During this exploration, if a node exceeds the current bound, its f(N) value is recorded as a candidate for the next bound. Among these, the bound is set to the minimum one, stored in minExceedingCost, and the algorithm proceeds to the next iteration with this new bound. At each depth level, the algorithm constructs the corresponding per-agent Directed Acyclic Graphs (DAGs) using Algorithm \ref{Algo4:DAG}, where each DAG $\mathcal{D}_i$ represents all feasible paths for agent $a_{i}$ of its currently allowed number of steps. These individual DAGs compactly store multiple path alternatives. Subsequently, a joint DAG $\mathcal{D} = \mathcal{D}_1 \times \ldots \times \mathcal{D}_\ell$ is constructed to represent all possible joint plans. To ensure feasibility, the joint DAG is pruned using the \textsc{PruneConflicts} procedure (in Algorithm \ref{Algo4:DAG}) and remaining plans are evaluated for fairness. All feasible and envy-free plans are accumulated in the set $\mathcal{R}$ over successive iterations. After the search terminates, fairness constraints, proportional fairness, and max–min fairness are applied to $\mathcal{R}$ to obtain the final set of fair and conflict-free joint plans. 

% Additionally, we can also conduct progress checks on the bound increments. If the difference between the current bound and the next minimum exceeding cost becomes smaller than a predefined threshold, then the search can be considered converged, and the algorithm is terminated.

\begin{algorithm}[h!]
\scriptsize
\SetAlgoLined
\caption{Depth-Bounded Search Procedure}
\label{Algo2:DepthBoundSearch}

\SetKwFunction{DepthBoundSearch}{DepthBoundSearch}
\SetKwProg{Fn}{Function}{:}{end}

\Fn{\DepthBoundSearch{$N$, depthBound, minExceedingCost, $\mathcal{R}$, $\epsilon$}}{

$f(N) \gets g(N) + h(N)$\;
\If{$f(N) >$ depthBound}{
    minExceedingCost $\gets \min($minExceedingCost$, f(N))$\;
    \Return \textbf{false}\;
}

Construct individual DAG $\mathcal{D}_i$ for each agent $i$ using Algorithm \ref{Algo4:DAG} corresponding to the number of steps in node $N$\;
Construct joint DAG $\mathcal{D} \gets \mathcal{D}_1 \times \cdots \times \mathcal{D}_\ell$\;
$\mathcal{D} \gets$ \textsc{PruneConflicts}($\mathcal{D}$)\;

\For{each joint plan $\Pi' = (\pi_1', \dots, \pi_\ell')$ in $\mathcal{D}$}{
    $C_i \gets |\pi_i'| \cdot c_i^{step}$\ {,} $W_i \gets u_i - C_i$\;
    \If{\textsc{IsEnvyFree}($\{W_i\}, \epsilon$)}{
        $\mathcal{R} \gets \mathcal{R} \cup \{\Pi'\}$\;
    }
}

\If{$\mathcal{R} \neq \phi$}{
    \Return \textbf{true}\;
}

found $\gets$ \textbf{false}\;
\For{each child node $N'$ in \textsc{ExpandNode}($N$)}{
    $g(N') \gets g(N) + \sum_i (s_i' - s_i) \cdot c_i^{step}$\;
    $h(N') \gets$ \textsc{HeuristicEstimate}($N'$)\;
    \If{\textsc{DepthBoundSearch}($N'$, depthBound, minExceedingCost, $\mathcal{R}$)}{
        found $\gets$ \textbf{true}\;
    }
}
\Return found\;}
\end{algorithm}

% Fairness Check Functions
\begin{algorithm}[t]
\scriptsize
\SetAlgoLined
\caption{Fairness Filtering Subroutines}
\label{Algo3:FairICTS_Fairness}

\SetKwFunction{IsEnvyFree}{IsEnvyFree}
\SetKwFunction{IsMaxMinFair}{IsMaxMinFair}

\SetKwProg{Fn}{Function}{:}{end}

\Fn{\IsEnvyFree{$\{W_i\}, \epsilon$}}{
    \For{each pair $(i, j)$, $i \neq j$}{
        \If{$|W_i - W_j| \ge \epsilon$}{
            \Return \textbf{false}\;
        }
    }
    \Return \textbf{true}\;
}

\vspace{0.5em}

\Fn{\IsMaxMinFair{$\Pi, \mathcal{R}$}}{
    $\mathbf{W} \gets$ sort$\big([W_1(\pi_1), W_2(\pi_2), \dots, W_\ell(\pi_\ell)]\big)$ 
    % \tcp*{Ascending order}
    
    \For{each agent $i = 1, 2, \dots, \ell$}{
        \For{each feasible solution $\Pi' \in \mathcal{R}$}{
            \If{$W_i(\pi'_i) > W_i(\pi_i)$ \textbf{and} $W_j(\pi'_j) \ge W_j(\pi_j),\ \forall j < i$}{
                \Return \textbf{false} \tcp*{Fairness violated}
            }
        }
    }
    \Return \textbf{true}\;
}
\vspace{0.5em}
\Fn{\IsProportionallyFair{$\Pi, \mathcal{R}$}}{
    \For{each $\Pi' \in \mathcal{R}$}{
        \If{$\displaystyle \sum_{i=1}^\ell 
            \frac{\mathcal{W}_i(\pi_i') - \mathcal{W}_i(\pi_i)}{\mathcal{W}_i(\pi_i)} > 0$}{
            \Return \textbf{false}\;
        }
    }
    \Return \textbf{true}\;
}
\end{algorithm}

\begin{algorithm}[h!]
\scriptsize
\SetAlgoLined
\caption{Build DAG for an Agent}
\label{Algo4:DAG}
\SetKwFunction{BuildDAG}{BuildDAG}
\SetKwFunction{PruneConflicts}{PruneConflicts}
\SetKwProg{Fn}{Function}{:}{end}

\Fn{\BuildDAG{$\mathcal{G}(\mathcal{V},\mathcal{E}), \tau_i=(a_i, v_{s_i}, v_{d_i}, u_i, c_i^{step}), s_i$}}{
    Build root node $(v_{s_i}, 0)$\;
    \For{each node $(v, t)$ in $DAG_i$ in BFS manner}{
        \If{$t < s_i$}{
            \For{$v' \in N(v) \cup \{v\}$}{
                Add edge $(v, t) \rightarrow (v', t+1)$\;
            }
        }
    }
    Prune leaf nodes $(v, t)$ if $t \neq s_i$ or $v \neq v_{d_i}$\;
    \Return $DAG_i$\;
}

\vspace{0.5em}

\Fn{\PruneConflicts{$\mathcal{D}$}}{
    % \tcc{Remove vertex and edge conflicts}

    \For{each node $n = (v_{a_1}, \dots, v_{a_\ell}, t)$ in $\mathcal{D}$}{
        \If{$\exists\, i \neq j$ such that $v_{a_i} = v_{a_j}$}{
            Remove $n$ from $\mathcal{D}$ \tcp*{Vertex conflict}
        }
    }

    \For{each edge $e = (u, v)$ in $\mathcal{D}$}{
        $u = (v_{a_1}, \dots, v_{a_\ell}, t)$; \quad $v = (v_{a_1}', \dots, v_{a_\ell}', t+1)$\;
        \For{$i \neq j \in \{1, \dots, \ell\}$}{
            \If{$v_{a_i} = v_{a_j}$ \textbf{and} $v_{a_i}' = v_{a_j}'$}{
                Remove $e$ \tcp*{Edge Conflict}
            }
            \ElseIf{$v_{a_i} = v_{a_j}'$ \textbf{and} $v_{a_j} = v_{a_i}'$}{
                Remove $e$ \tcp*{Swap conflict}
            }
        }
    }
    \Return $\mathcal{D}$\;
}
\end{algorithm}

\begin{figure}[h!]
    \centering
    \includegraphics[width=0.2\columnwidth]{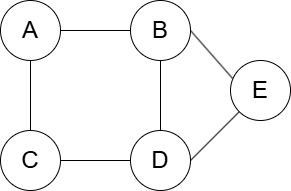}
    \includegraphics[width=0.2\columnwidth]{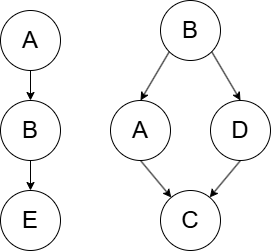}
    \includegraphics[width=0.2\columnwidth]{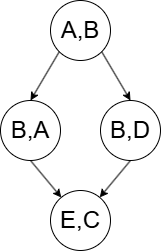}
    \caption{
    \scriptsize
    (a) Environment graph $\mathcal{G}$, 
             (b) individual agent DAGs $\mathcal{D}_1$ and $\mathcal{D}_2$, and 
             (c) their joint DAG $\mathcal{D} = \mathcal{D}_1 \times \mathcal{D}_2$, 
             with step bound 2.}
    \label{fig:example_dag}
\end{figure}

\subsection{Fair Conflict-Based Search}
The Fair Conflict-Based Search (Fair-CBS) algorithm extends the standard CBS framework to ensure fairness. The algorithm begins by computing independent shortest paths for all agents using $A^*$, forming the root node of the search tree. Each node contains a joint plan, associated constraints, and a sum of individual path costs (SIC) representing the total cost. At each iteration, the algorithm expands the node with the lowest cost from the open list. If the current solution is conflict-free, it is tested for envy-freeness. Envy-free solutions are stored in $\mathcal{R}$. The search continues until all nodes have been expanded, ensuring that every feasible leaf node is evaluated and stored if it is a valid solution. When a conflict is detected between two agents $(a_i, a_j)$ at location $v$ and time $t$, two child nodes are generated by imposing mutually exclusive constraints that prevent one of the agents from occupying $(v,t)$. For each child, the affected agent’s path is replanned via $A^*$ considering the updated constraints, and the resulting node is added to the search frontier. After the search tree is fully expanded, the set $\mathcal{R}$ is further filtered to retain only those joint plans that satisfy proportional fairness and max–min fairness. Finally, the algorithm returns the welfare-maximizing plan.

\begin{algorithm}[h!]
\scriptsize
\SetAlgoLined
\SetKwFunction{IsProportionallyFair}{IsProportionallyFair}
\SetKwFunction{IsMaxMinFair}{IsMaxMinFair}
\SetKwFunction{IsEnvyFree}{IsEnvyFree}
\caption{Fair Conflict-Based Search}
\label{Algo:FairCBS}

\KwIn{Graph $\mathcal{G} = (\mathcal{V}, \mathcal{E})$, 
Agents $\{ \tau_i = (a_i, v_{s_i}, v_{d_i}, u_i, c_i^{step}) \}_{i=1}^{\ell}$, 
Fairness tolerance $\epsilon$}
\KwOut{Fair and conflict-free joint plan $\Pi^\ast$}

$\mathcal{R} \gets \phi$ 

Compute individual paths $\pi_i \gets$ \textsc{$A^*$} $(v_{s_i}, v_{d_i})$ for all $i$, $Root.\text{solution} = (\pi_1, \dots, \pi_\ell)$\;
$Root.\text{cost} \gets \text{SIC}(Root.\text{solution})$\; $Root.\text{constraints} \gets \phi$\;
Insert $Root$ into \textsc{OPEN}\;

\While{\textsc{OPEN} is not empty}{
    $P \gets$ remove a node from \textsc{OPEN}\;
    $conflict \gets$ \textsc{FindConflict}$(P.\text{solution})$\;

    \If{$conflict = \phi$}{
        \If{\textsc{IsEnvyFree}$(\{W_i\} \text{ from } P.\text{solution}, \epsilon)$}{
            $\mathcal{R} \gets \mathcal{R} \cup \{P.\text{solution}\}$\;
        }
        \textbf{continue}\;
    }

    Let $conflict = (a_i, a_j, v, t)$\;
    \For{$a \in \{a_i, a_j\}$}{
        Create child node $A$ with constraints $A.\text{constraints} = P.\text{constraints} \cup \{(a, v, t)\}$\;
        Replan $\pi_a$ using \textsc{$A^*$}$(v_{s_a}, v_{d_a}, A.\text{constraints})$\;
        \If{$\pi_a \neq \phi$}{
            Update $A.\text{solution}$ and $A.\text{cost} = \text{SIC}(A.\text{solution})$\;
            Insert $A$ into \textsc{OPEN}\;
        }
    }
}
$\mathcal{R} \gets \{\Pi \in \mathcal{R} \mid$ \IsProportionallyFair{$(\Pi, \mathcal{R})$} $\land$ \IsMaxMinFair{$(\Pi, \mathcal{R})$}$\}$\;
\Return $\Pi^\ast = \arg\max_{\Pi \in \mathcal{R}} \sum_i W_i(\pi_i)$\;

\end{algorithm}

\subsection{Fairness Verification}
Our approach integrates fairness-aware evaluation directly into the search. Max-min fairness is enforced by ordering agents in non-decreasing welfare and discarding any plan for which another feasible plan can strictly improve the welfare of some agent $i$ without improving the welfare of agents $1,\dots,i-1$. Next, proportional fairness is checked to ensure that no alternative plan can increase total welfare proportionally without inducing a larger proportional loss to others. Only fair plans are retained in $\mathcal{R}$, and the final solution is selected as $\Pi^* = \arg\max_{\Pi \in \mathcal{R}} \sum_{i=1}^{\ell} W_i(\pi_i)$.

Next, we establish the theoretical guarantees of the proposed Fair-ICTS algorithm.

\begin{theorem}
Let $\mathcal{F}_{fair}$ denote the set of all fair and feasible joint plans. If there exists at least one feasible MAPF solution, then the Fair-ICTS algorithm
terminates and returns an optimal fair joint plan $\Pi^\ast = \arg\max_{\Pi \in \mathcal{F}_{fair}} \sum_{i=1}^{\ell} W_i(\pi_i)$.

\end{theorem}

\begin{proof}
Fair-ICTS performs a depth-bound search, where the search bound is iteratively increased. At each iteration, it explores all nodes $N$ with $f(N) \leq \text{depthBound}$. For a given step vector, the corresponding DAGs are finite due to finite map and time horizon. Hence, the joint DAG enumerates all feasible joint plans for that bound. The pruning procedure removes only nodes and edges that violate conflict constraints. Therefore, if a valid solution $\Pi$ exists, there is a finite bound $B^\ast$ exceeding its cost, and algorithm will eventually generate $\Pi$ once the bound reaches $B^\ast$. This establishes completeness. By the properties of $IDA^*$, search bounds are increased monotonically and all nodes with finite cost are eventually explored. Consequently, every feasible plan is generated at some iteration. After filtering by fairness constraints, the remaining set of plans is exactly $\mathcal{F}_{fair}$. Finally, Fair-ICTS selects and returns the welfare maximizing plan, which establishes optimality.
\end{proof}

\begin{theorem} The time and space requirements of Fair-ICTS are $\mathcal{O}(B \cdot b^{s_{max}+\ell} \cdot \ell^2\;+\;B \cdot P_{\text{cf}}^2 \cdot \ell
\;+m+nlog(n))$ and $\mathcal{O}(b^{\ell + s_{\max}})$ respectively, where $B$, b, $s_{max}$, $P_{cf}$ denote number of $IDA^*$ iterations, maximum branching factor of DAG, maximum number of steps and number of conflict-free joint plans at any bound.
\end{theorem}
% \begin{proof}
% Firstly, the optimal number of steps is to be calculated for each agent using the $A^*$ algorithm, which can be done in $\mathcal{O}(m+nlog(n))$ time. For a node $(s_1, s_2, \dots, s_\ell)$ in the search space, let $s_{max} = max_i s_i$ be the maximum element. Let $\Delta$ be the maximum vertex degree in the graph $\mathcal{G}(\mathcal{V}, \mathcal{E})$ and $b$ denote the maximum branching factor of the DAG for any agent. Then, $b \leq \Delta +1$. So, construction of $DAG_i$ requires $\mathcal{O}(b^{s_i})$ = $\mathcal{O}(b^{s_{max}})$ time. The joint DAG has $N_{\text{joint}} = \mathcal{O}(b^{s_{max}})$ nodes and $\mathcal{O}(N_{\text{joint}} \cdot b^\ell)$ edges. Conflict pruning requires pairwise collision checks among agents, incurring a cost of $\mathcal{O}(N_{\text{joint}} \cdot b^\ell \cdot \ell^2)$ in the worst case. After pruning, up to $P_{\text{cf}}$ conflict-free joint plans may be enumerated at the current cost bound. The fairness filtering phase requires $\mathcal{O}(P_{\text{cf}}^2 \cdot \ell)$ time. Let $B$ denote the number of  $IDA^*$ iterations. The  worst-case time complexity of the algorithm is thus 
% $
% \mathcal{O}\Big(
% B \cdot b^{s_{max}+\ell} \cdot \ell^2
% \;+\;
% B \cdot P_{\text{cf}}^2 \cdot \ell
% \;+
% m+nlog(n) \Big)
% $.

% The dominant space requirement of Fair-ICTS arises from storing the joint DAG, that is $\mathcal{O}(b^{\ell + S_{\max}})$.
% \end{proof}

\begin{theorem}
Fair-CBS finds at least one $\Pi\in\mathcal{F}_{fair}$, if $\mathcal{F}_{fair}\neq\phi$ and returns plan with maximimum social welfare over $\mathcal{F}_{fair}$.
\end{theorem}

\begin{proof}
Since Fair-CBS exhaustively generates the CBS search tree where every feasible plan corresponds to at least one leaf. Plans corresponding to leaf nodes are checked for fairness constraints and saved in the set $\mathcal{R}$. Thus $\mathcal{F}_{fair} = \mathcal{R}$ and welfare maximizing plan is returned.
% this proves completeness. For optimality, the plan with maximum welfare is returned.
\end{proof}

\begin{theorem} The time and space requirements of Fair-CBS are $\mathcal{O}(N_{\text{CT}}\big(m + n\log n + \ell^2\big) \;+\; N_{\text{CT}}^2\cdot \ell^2)$ and $\mathcal{O}(N_{\text{CT}}\cdot \ell\cdot T_{\max})$ respectively, where $N_{CT}$ and $T_{max}$ are number of nodes explored and space to store an agent's path repectively.
\end{theorem}
% \begin{proof}
% The root node consists of $A^*$ paths for all agents whose computation requires $\mathcal{O}\Big(N_{\text{CT}}\cdot\big(m + n\log n\big)\Big)$ time considering that Fair-CBS expands $N_{\text{CT}}$ nodes. Let L be the number of conflict-free leaves produced. Checking of envy-freeness requires $\mathcal{O}(L\cdot \ell^2)$ time over all leaves, and the post-filtering for max–min and proportional fairness, which involves pairwise comparisons, can be bounded by $\mathcal{O}(L^2\cdot \ell^2)$. $L = \mathcal{O}(N_{\text{CT}})$, so worst case time complexity of the algorithm is $\mathcal{O}\Big(N_{\text{CT}}\big(m + n\log n + \ell^2\big) \;+\; N_{\text{CT}}^2\cdot \ell^2\Big)$.

% Each search tree node stores a joint plan. Consider storing a single agent path requires at most \(\mathcal{O}(T_{\max})\) memory, hence storing \(N_{\text{CT}}\) nodes requires $\mathcal{O}\big(N_{\text{CT}}\cdot \ell\cdot T_{\max}\big)$ space.
% \end{proof}
\vspace{-0.25 in}
\section{Mechanism for Rational Agents}\label{Sec:mechanism_design}
In this section, we assume that the per-step cost of the agent ($c_i^{step}$) is private and is reported by the agent to the system. All other information is public and known apriori to the system. Algorithm \ref{Algo1:FairICTS} and \ref{Algo:FairCBS} compute a set of fair and non-conflicting solutions ($\mathcal{R}$) and select a plan that maximizes social welfare. However, the agents may be self-interested and misreport their costs if doing so is beneficial to them. We propose a mechanism with the following properties.
\begin{itemize}
\item \textbf{Dominant Strategy Incentive compatible(DSIC).} For each agent, truthfully reporting their private information is the best strategy, regardless of what other agents do.
\item \textbf{Individually Rational(IR).} Each agent is at least as well off participating in the mechanism as they would be by not participating.
\item \textbf{Performance Guarantee.} Given the reported types, the mechanism outputs the solution with maximum social welfare.
\end{itemize}

We begin by defining some preliminary concepts. 
% First, we define the single-parameter domain mechanism in Definition \ref{Def:SPDM}.

\begin{definition}[Single Parameter Domain Mechanism]\label{Def:SPDM}
Let each agent report her bid $b_i$, claiming it as her true cost. Given reported bids \textbf{b} = $(b_1,b_2,\dots,b_{\ell})$, a single-parameter domain mechanism denoted by $(f,p)$ can be described by the following two parts:
\begin{itemize}
\item \textbf{Allocation function ($f$)} outputs a joint plan from the set of plans (f($b$) $\in \mathcal{R}$).
\item \textbf{Payment function ($p$)} gives payment for each agent ($p_i(b) \in \mathbb{R}$).
\end{itemize}
\end{definition}
% Next, we define the winning set in Definition \ref{Def:WS}.
\begin{definition}[Winning Set]\label{Def:WS}
Let $\pi_i^* = f(c_1,\dots,c_l)_i$ be agent $i$'s assigned path under true costs, then the winning set for agent i is defined as: $\mathcal{W}_i = \{\Pi \in \mathcal{R} | \pi_{i} = \pi_{i}^* \}$, i.e., the set of all plans where agent i gets the same path as in the optimal solution under true costs.
\end{definition}
% Next, we define the monotonicity of the allocation function $f$ in Definition \ref{Def:Monotone}.
\begin{definition}[Monotonicity]\label{Def:Monotone}
Given bids $(b_i,b_{-i})$, the allocation function $f$ is monotone for agent i if the following condition holds true:
$f(b_i,b_{-i}) \in \mathcal{W}_i \implies f(b_i',b_{-i}) \in \mathcal{W}_i$ $\forall~ b_i' \geq b_i$.
\end{definition}
% Next, we define the critical value in Definition \ref{Def:Critical_Value}.
\begin{definition}[Critical Value]\label{Def:Critical_Value}
The critical value for agent i for a monotone allocation function $f$ is defined as:   $r_i(b_{-i}) = sup\{b_i | f(b_i,b_{-i}) \notin \mathcal{W}_i\}$. Thus, an agent wins by bidding greater than its critical value. Hence, the critical value can be computed as: $r_i(b_{-i}) = inf\{ \forall \epsilon>0, f(b_i,b_{-i})_i = f(b_i+\epsilon,b_{-i})_i \}$. For fixed $b_{-i}$, the mechanism computes $r_i(b_{-i})$ as the smallest bid $b_i$ such that the allocation $f(b_i,b_{-i})_i$ differs from that obtained for any smaller bid, which can be found by a one-dimensional search over $b_i$.
\end{definition}

% \subsection{Mechanism Design for Fair MAPF}
\paragraph{Mechanism Settings}
% The graph $\mathcal{G} = (\mathcal{V}, \mathcal{E})$, including the source, destination, and value received upon reaching the goal for each agent, is public and known to the mechanism. 
Each agent reports her per-step cost to the mechanism, and then the mechanism selects a solution from $\mathcal{R}$. We denote the total number of steps in path $\pi_i$ of agent $i$ as $k_i(\pi_i)$. Each agent aims to maximize her utility, $U_i = max\{u_i - c_i \times k_i(\pi_i) - p_i, 0\}$, where $p_i$ is the payment charged from agent i. The goal of the mechanism is to maximize social welfare under reported bids, $W(\textbf{b};\Pi) = \Sigma_{i=1}^l (u_i - b_i \times k_i(\pi_i))$, where $\textbf{b} = (b_1, b_2, \dots, b_{\ell})$.

\paragraph{Mechanism.}
We propose a single-parameter domain mechanism $(f,p)$ that maximizes social welfare and can be described as:
\begin{itemize}
% \item Each agent reports bid $b_i \in \mathbb{R}^+$ as her claimed cost.
\item Allocation function, $f: (\mathbb{R}^+)^{\ell} \to \mathcal{R}$ is defined as $f(b_1,b_2,\dots,b_{\ell}) = \argmax_{\Pi \in \mathcal{R}}$ $\sum_{i=1}^{\ell} (u_i - b_i \times k_i(\pi_i))$.
\item Payment function, $p_i: (\mathbb{R^+)^{\ell}} \to \mathbb{R} \cup \{0\}$ is defined as $p_i(\textbf{b}) = r_i(b_{-i}) \times k_i(f(\textbf{b})_i)$. 
\end{itemize} 

\begin{theorem} \label{th3: DSIC_IR}
The above proposed single-parameter domain mechanism $(f,p)$ is DSIC and IR.
\end{theorem}

\section{Experimental Evaluation}\label{Sec:Exp_Evualuation}
This section describes the experimental setup. Initially, we discuss the datasets used in our experiments.

\subsection{Dataset Description}
We conduct our experiments on four 2D grid maps selected from commonly used MAPF benchmarks by Stern et al. \cite{stern2019multi}. These maps vary in size and structure. Empty maps contain no obstacles, allowing free movement across all cells, whereas other maps include static obstacles that restrict paths. The complete list of datasets is provided in Table \ref{tab:maps}.
% and visualization of some maps is shown in Figure \ref{fig:maps_visualization}.

% \begin{figure}[t]
%     \centering
%     \includegraphics[width=0.3\columnwidth]{Images/grid_random32-32.png}
%     \includegraphics[width=0.3\columnwidth]{Images/grid_den312d.png}
%     \caption{(a) \texttt{random32-32-20}
%              (b) \texttt{den312d} \\
%             White cells are free for movement and black cells represent obstacles.}
%     \label{fig:maps_visualization}
% \end{figure}

\begin{table}[h!]
\centering
\scriptsize
\caption{Benchmark MAPF Datasets Used in Experiments}
\label{tab:maps}
\begin{tabular}{l c}
\hline
\textbf{Map} & \textbf{Size} \\
\hline
% empty-4-4 & $4 \times 4$ \\
empty-16-16 & $16 \times 16$ \\
random-32-32-20 & $32 \times 32$ \\
empty-48-48 & $48 \times 48$ \\
den312d & $81 \times 65$ \\
% den520d & $257 \times 256$ \\
\hline
\end{tabular}
\end{table}
\vspace{-0.2 in}
\subsection{Key Parameters}
Each experiment was executed 100 times with different randomly selected start and goal positions for all agents. For agent \( a_i \), the utility and the cost of the step were sampled as \emph{$ u_i \sim \text{Uniform}(0.001, 1) $} and \emph{$c_i^{\text{step}} \sim \text{Uniform}\big(10^{-6}, \max(10^{-6}, \frac{u_i}{\text{dist}_i})\big)$}, where \(\text{dist}_i\) is the optimal path length of the agent. For experimental evaluation, each algorithm was executed with a time limit of 60 seconds per instance. The \textit{success fraction} is defined as the ratio of successful runs (i.e., instances where a solution was found within the time limit) to the total number of runs. The average runtime over 100 runs is reported for each setting.

\subsection{Goals of our Experiments}
In this study, we want to address the following Research Questions (RQ).
\begin{itemize}
\item \textbf{RQ1}: How do the Fair-ICTS and Fair-CBS algorithms perform under different fairness notions?
    
\item \textbf{RQ2}: How does the scalability of both algorithms vary with respect to the number of agents, in terms of success fraction and runtime performance?

\item \textbf{RQ3}: How does varying the $\epsilon$ parameter influence the trade-off between fairness satisfaction, success rate, and runtime for both algorithms?
\end{itemize}
% To answer these questions, we systematically vary the number of agents and fairness parameters across multiple benchmark maps, reporting the resulting success fractions and run times for each configuration.
\subsection{Experimental Results}

% \begin{figure*}[!ht]
% \centering
% \begin{tabular}{cccc}
% \includegraphics[scale=0.13]{graphs_pdfs/random32-32_prop_sf.pdf} & 
% \includegraphics[scale=0.13]{graphs_pdfs/random32-32_prop_runtime.pdf} & 
% \includegraphics[scale=0.13]{graphs_pdfs/random32-32-envy-runtime.pdf} & 
% \includegraphics[scale=0.13]{graphs_pdfs/random32-32_max-min_sf.pdf} \\
% \tiny{(a)} & \tiny{(b)} & \tiny{(c)} & \tiny{(d)} \\[4pt]

% \includegraphics[scale=0.13]{graphs_pdfs/random32-32_max-min_runtime.pdf} & 
% \includegraphics[scale=0.13]{graphs_pdfs/empty48-48_prop_sf.pdf} & 
% \includegraphics[scale=0.13]{graphs_pdfs/empty48-48_prop_runtime.pdf} & 
% \includegraphics[scale=0.13]{graphs_pdfs/empty48-48_max-min_sf.pdf} \\
% \tiny{(e)} & \tiny{(f)} & \tiny{(g)} & \tiny{(h)} \\[4pt]

% \includegraphics[scale=0.13]{graphs_pdfs/empty48-48_max-min_runtime.pdf} & 
% \includegraphics[scale=0.13]{graphs_pdfs/den312d_prop_runtime.pdf} & 
% \includegraphics[scale=0.13]{graphs_pdfs/empty16-16_max-min_sf.pdf} & 
% \includegraphics[scale=0.13]{graphs_pdfs/empty16-16_max-min_runtime.pdf} \\
% \tiny{(i)} & \tiny{(j)} & \tiny{(k)} & \tiny{(l)} \\[4pt]
% \end{tabular}

\begin{figure*}[!ht]
\centering
\begin{tabular}{ccc}

% Row 1
\includegraphics[scale=0.172]{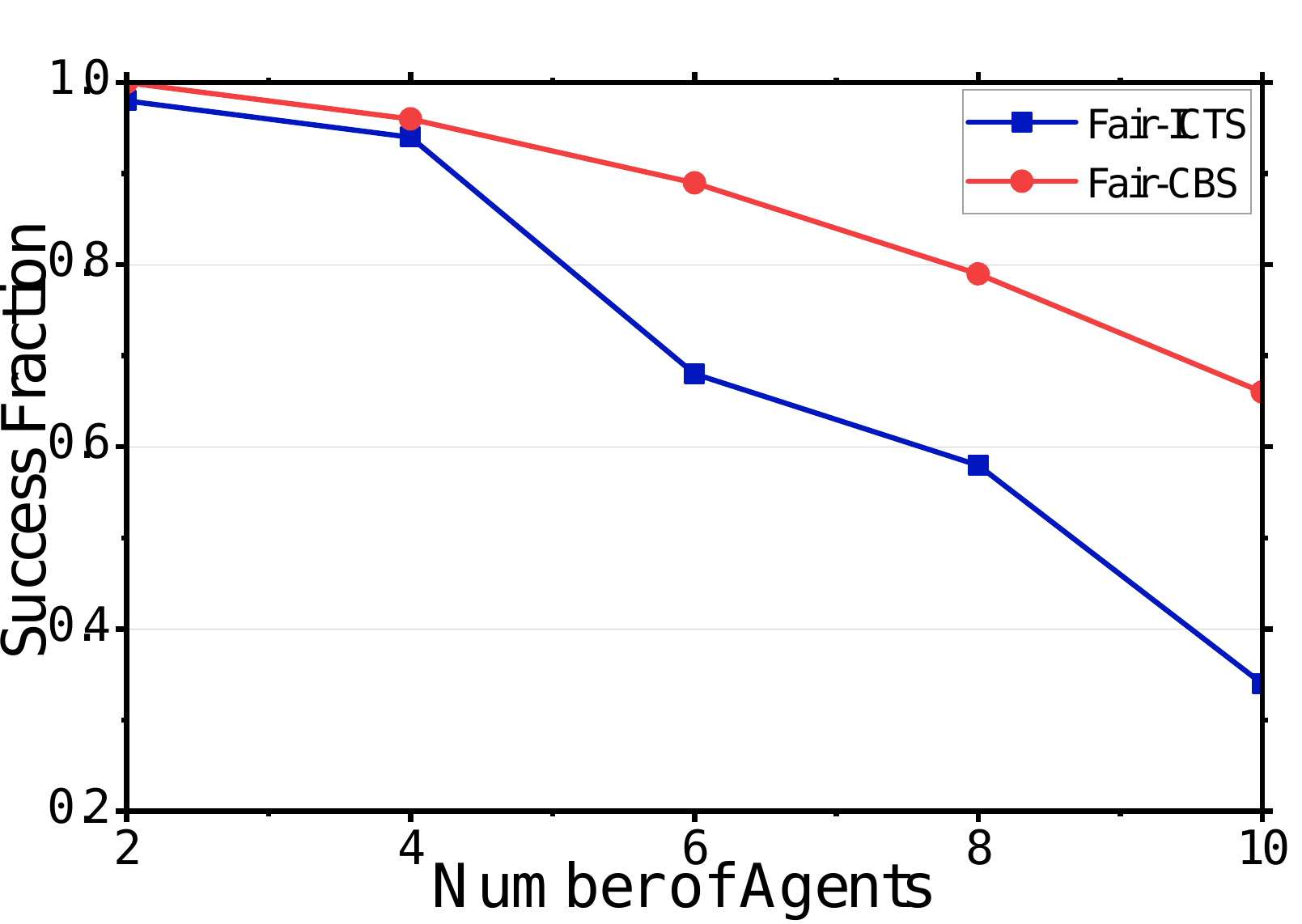} &
\includegraphics[scale=0.172]{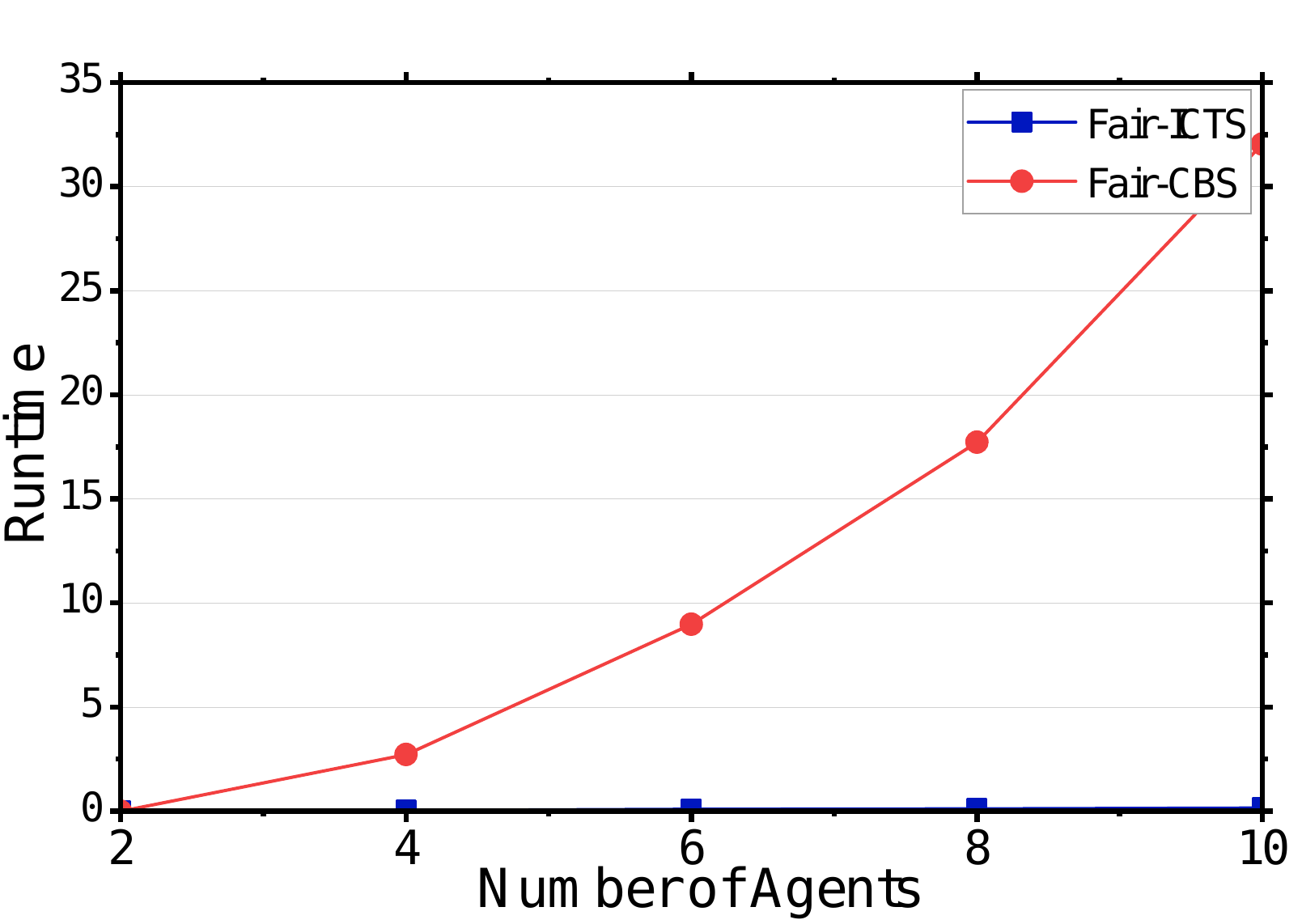} &
\includegraphics[scale=0.172]{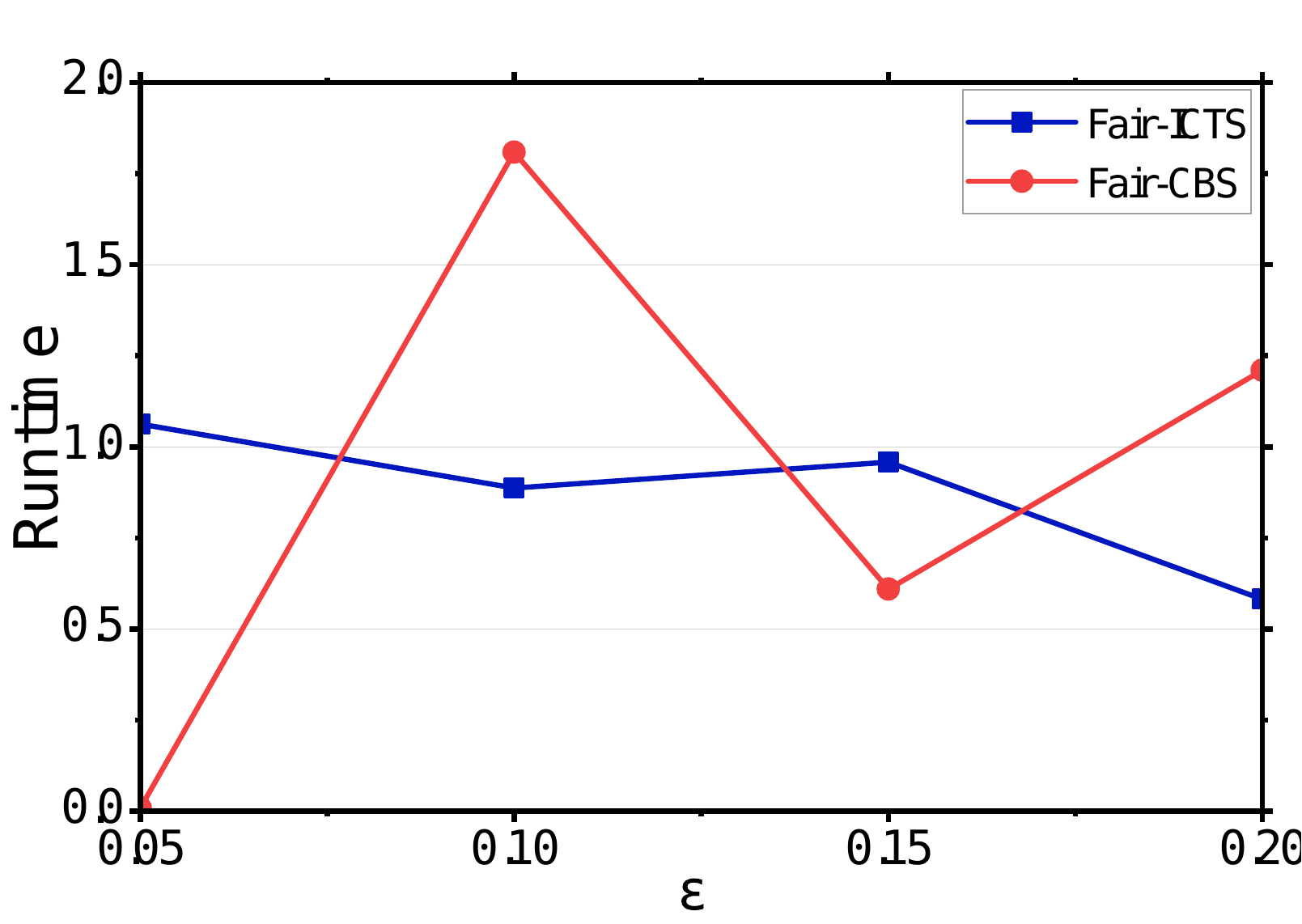} \\
\tiny{(a)} & \tiny{(b)} & \tiny{(c)} \\[4pt]

% Row 2
\includegraphics[scale=0.172]{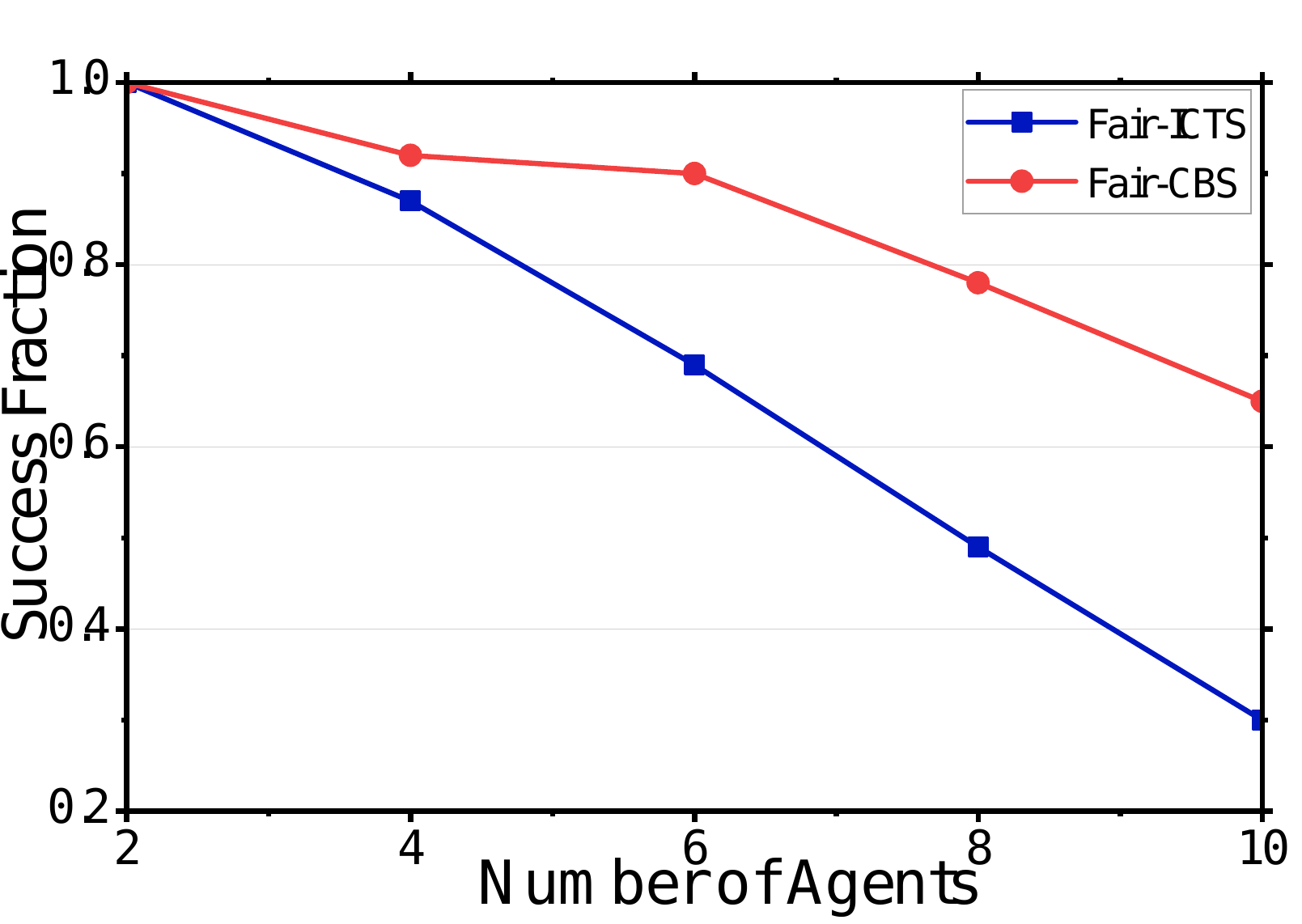} &
\includegraphics[scale=0.172]{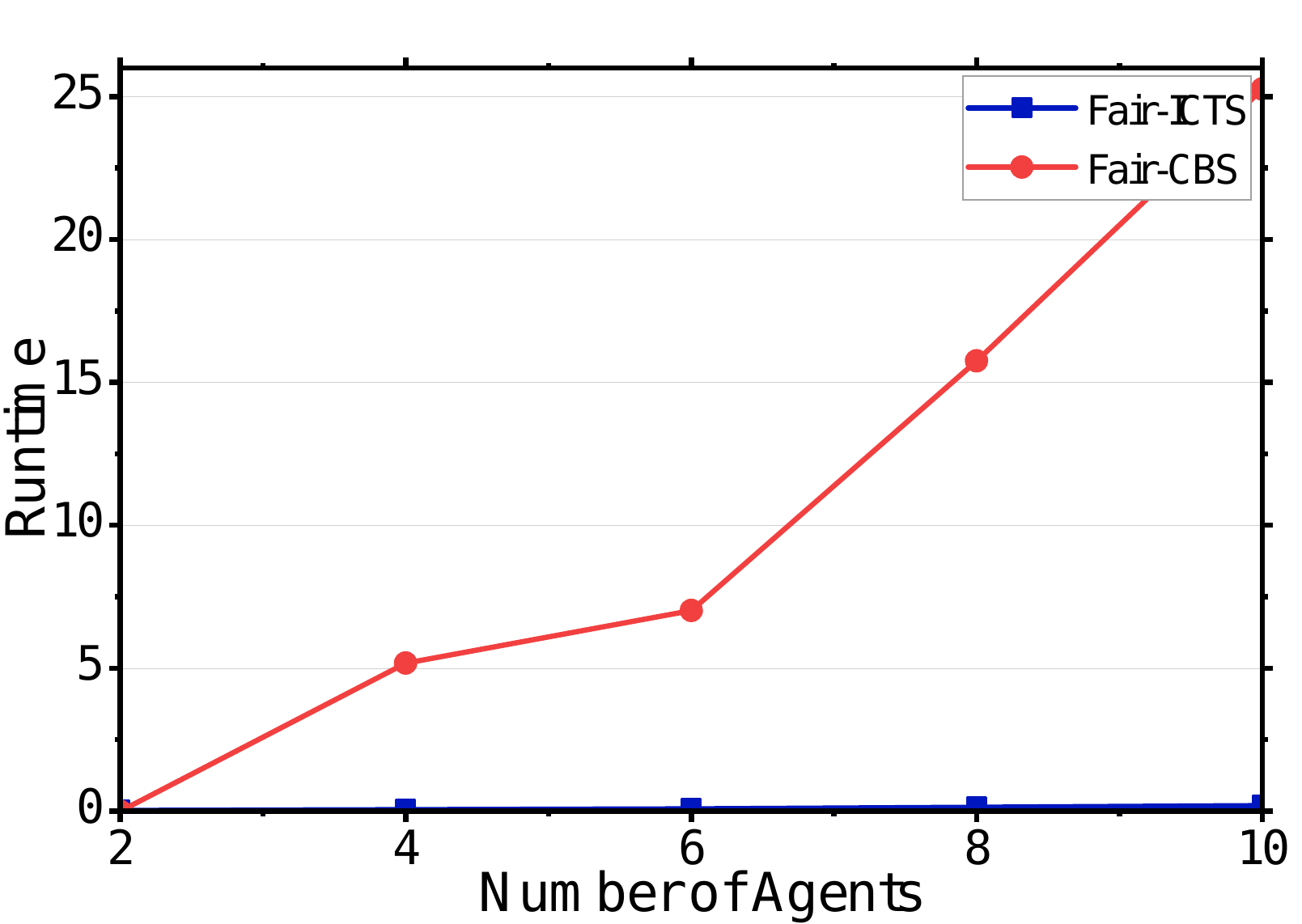} &
\includegraphics[scale=0.172]{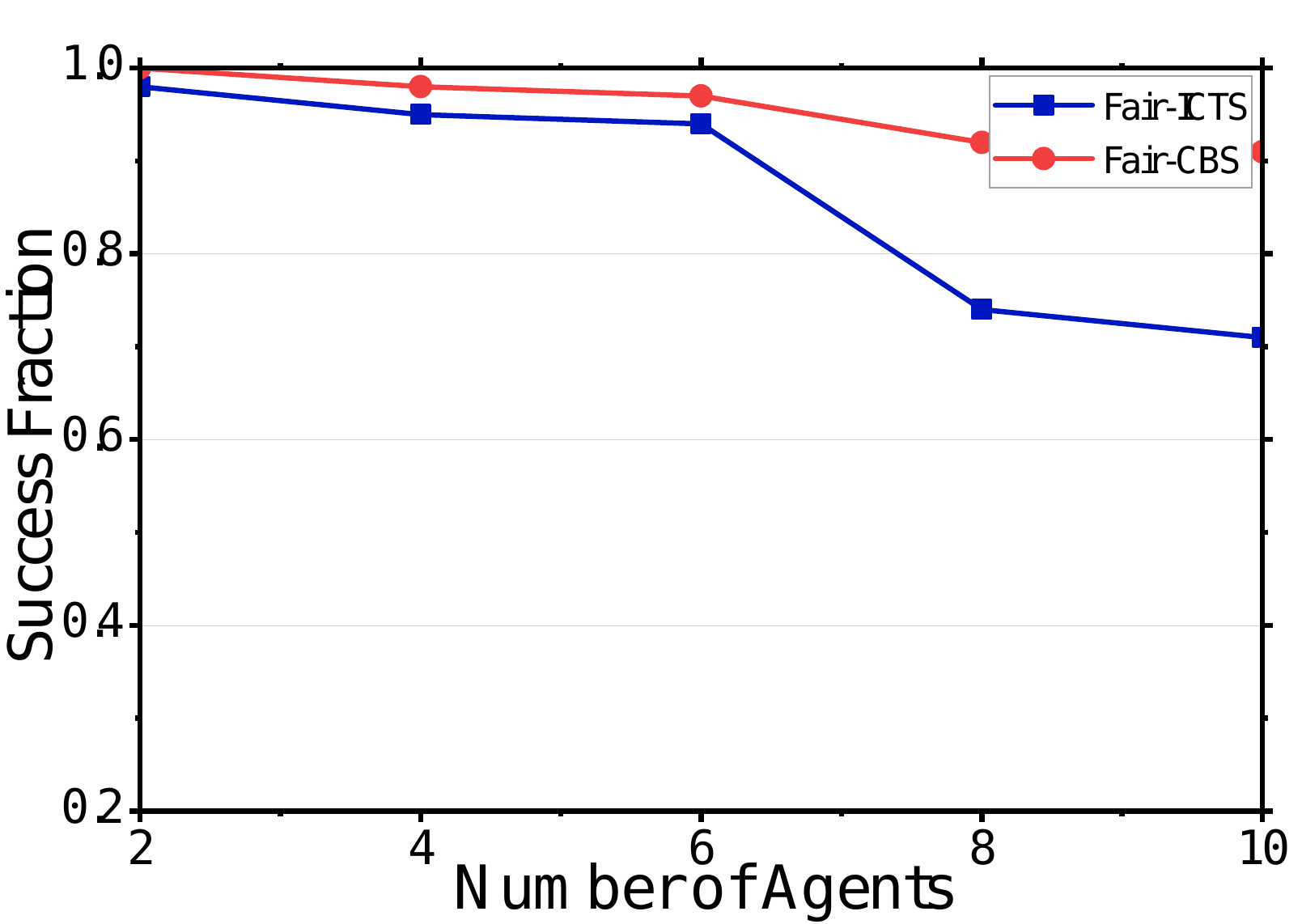} \\
\tiny{(d)} & \tiny{(e)} & \tiny{(f)} \\[4pt]

% Row 3
\includegraphics[scale=0.172]{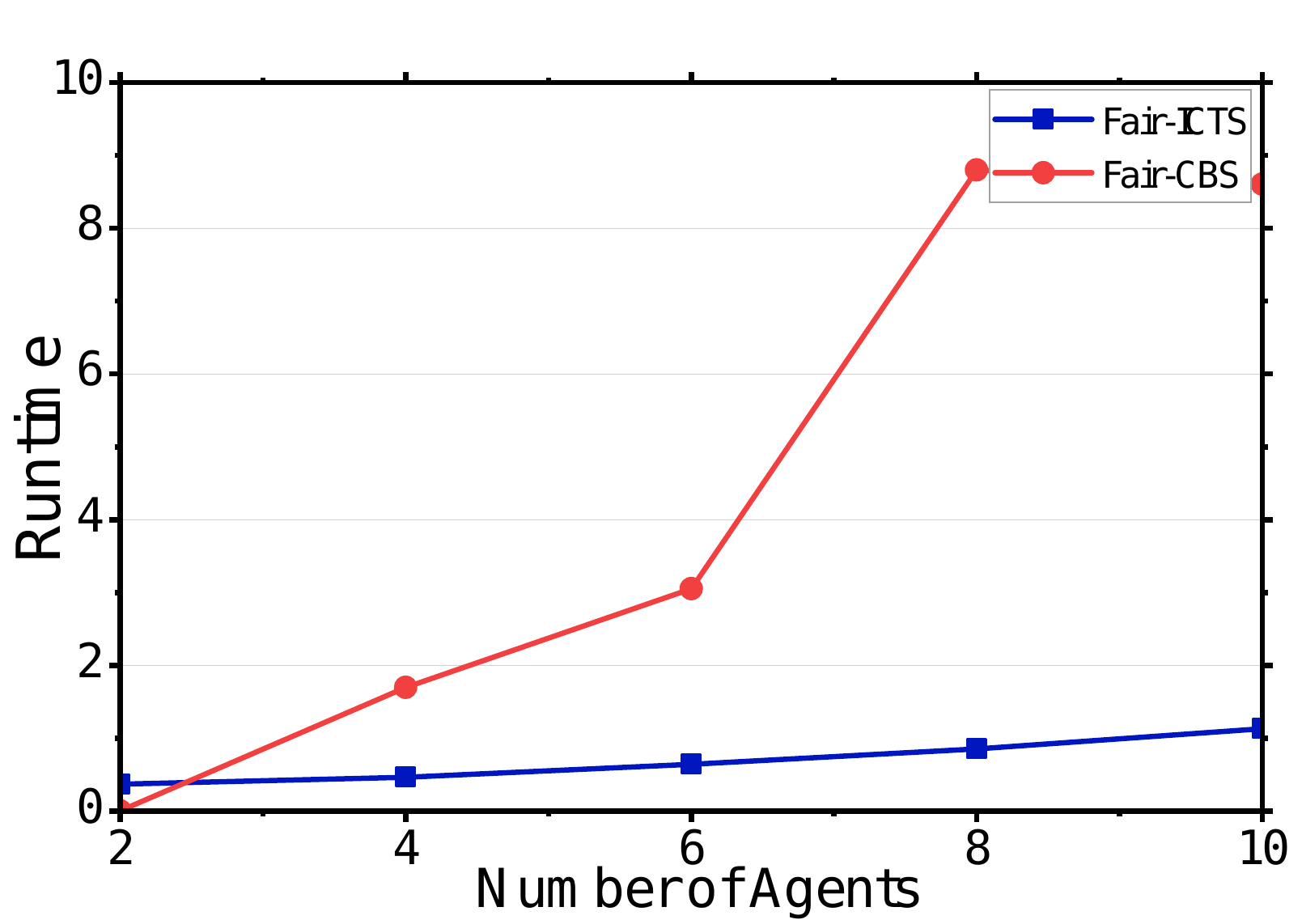} &
\includegraphics[scale=0.172]{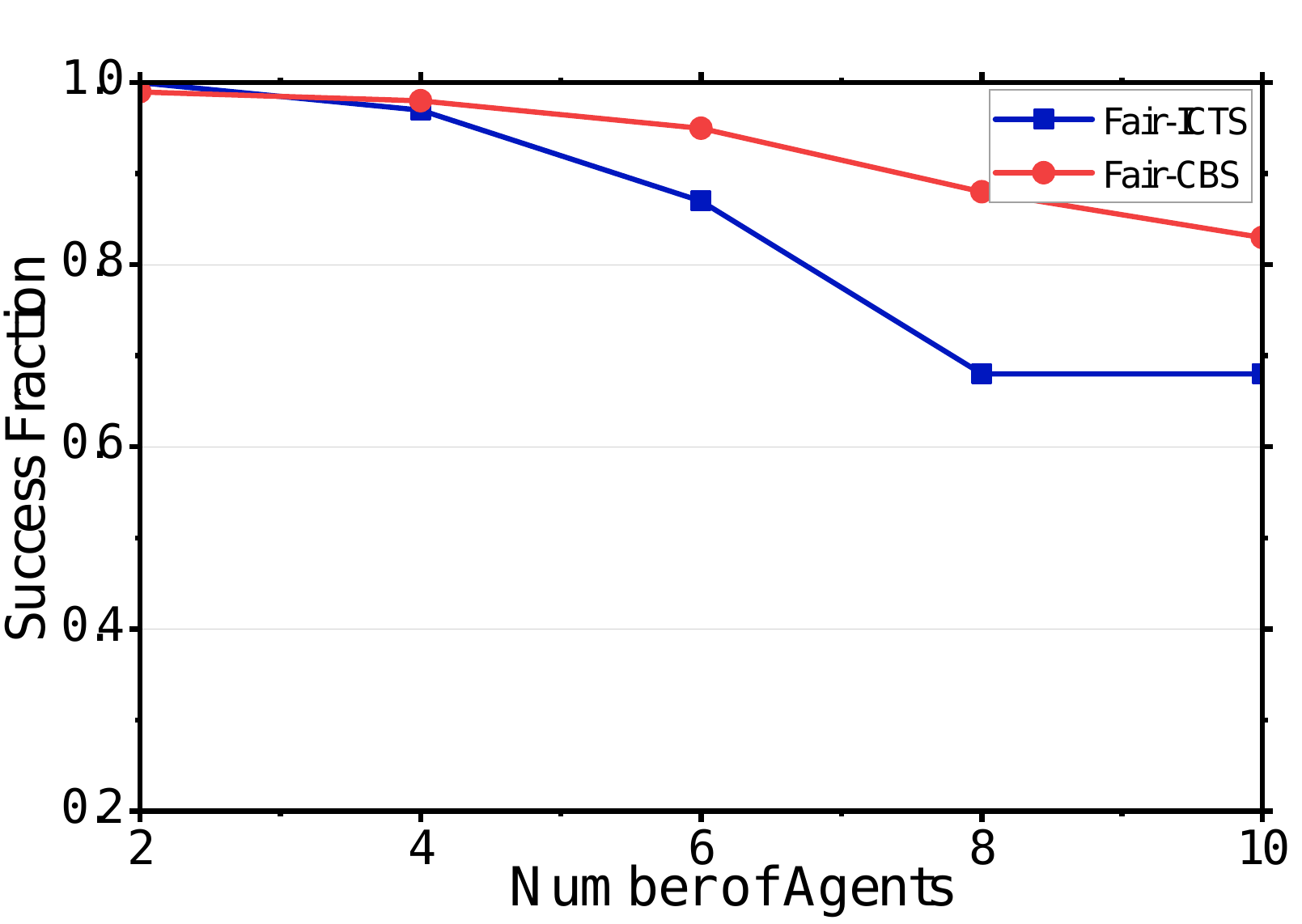} &
\includegraphics[scale=0.172]{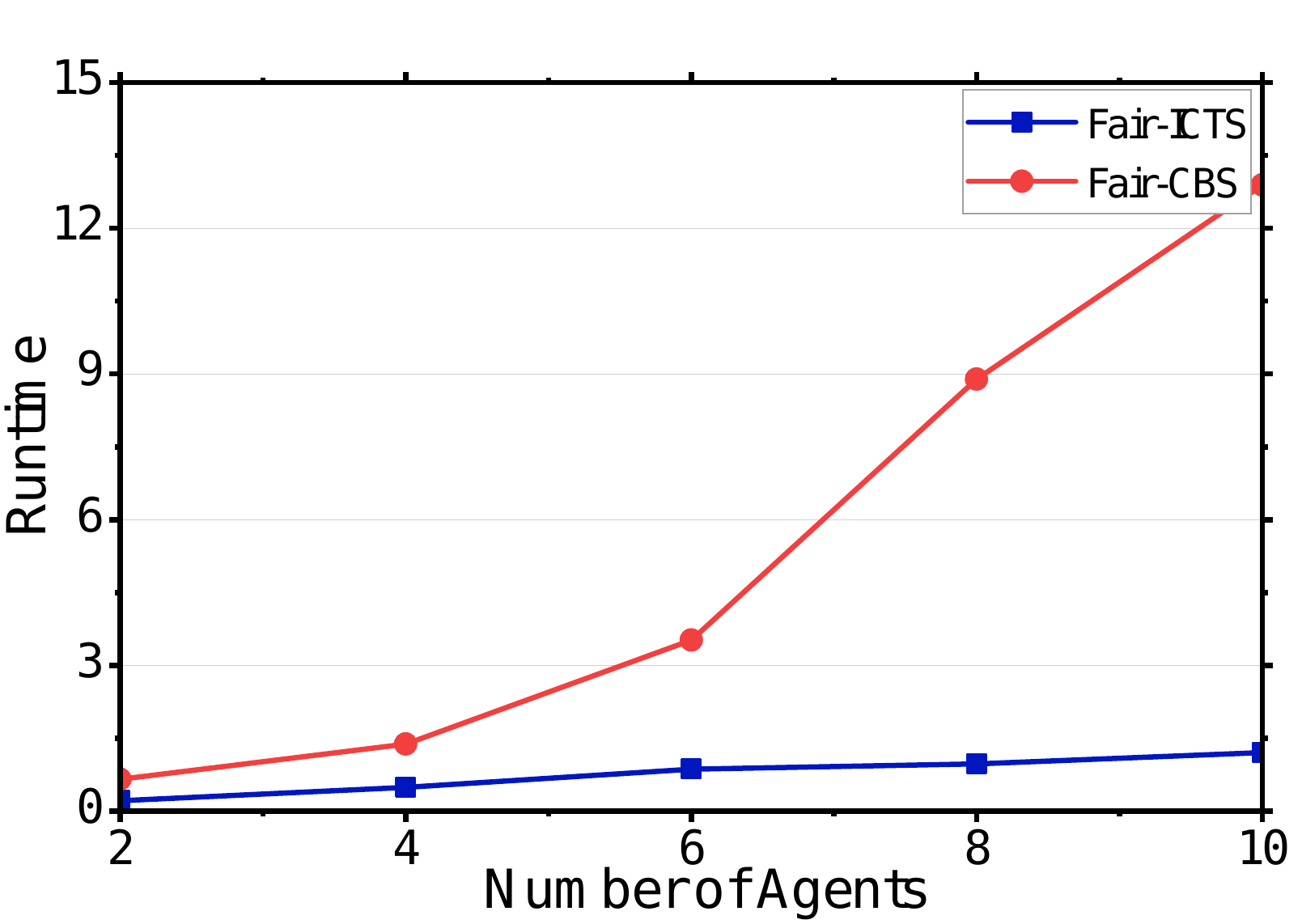} \\
\tiny{(g)} & \tiny{(h)} & \tiny{(i)} \\[4pt]

% Row 4
\includegraphics[scale=0.172]{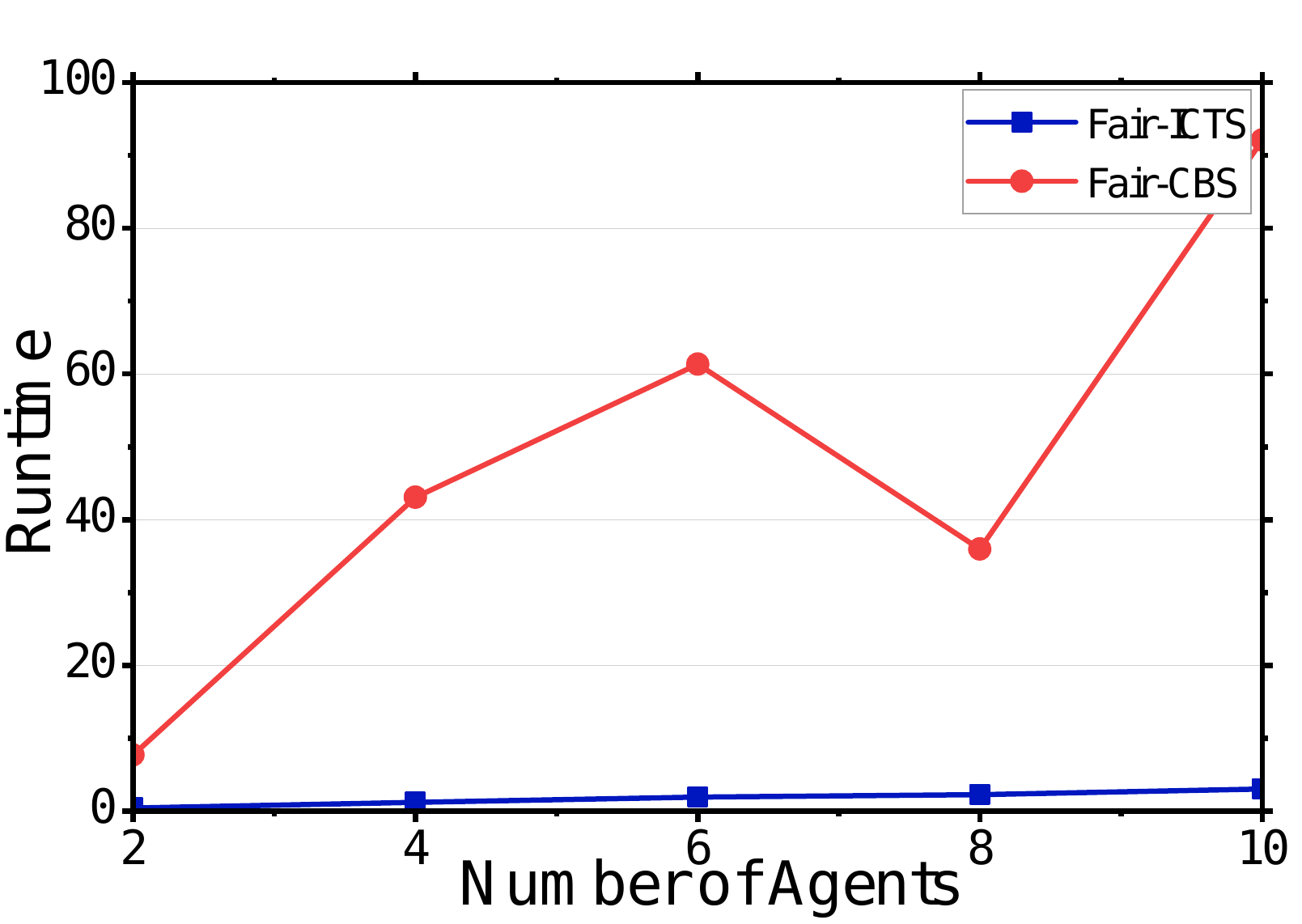} &
\includegraphics[scale=0.172]{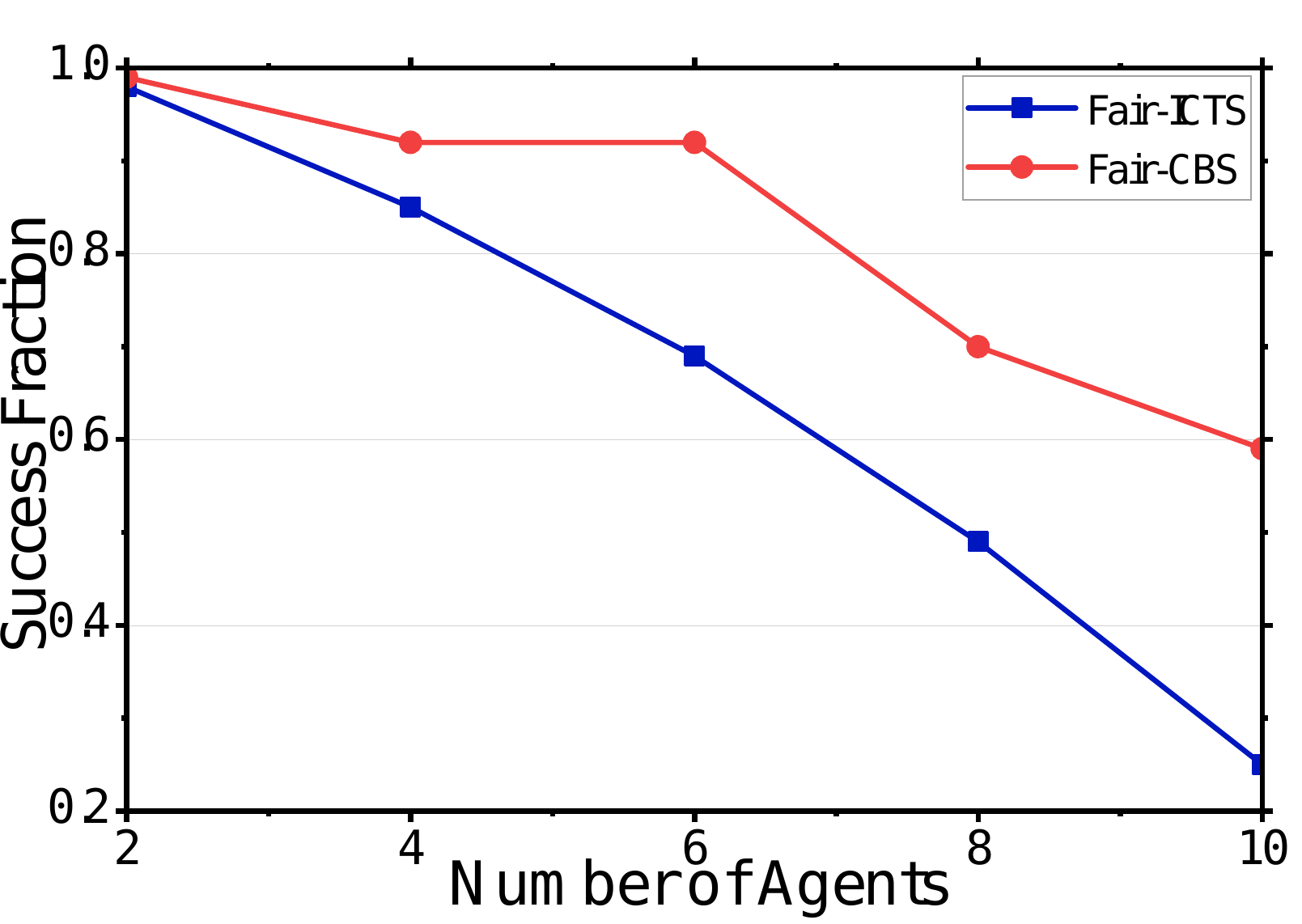} &
\includegraphics[scale=0.172]{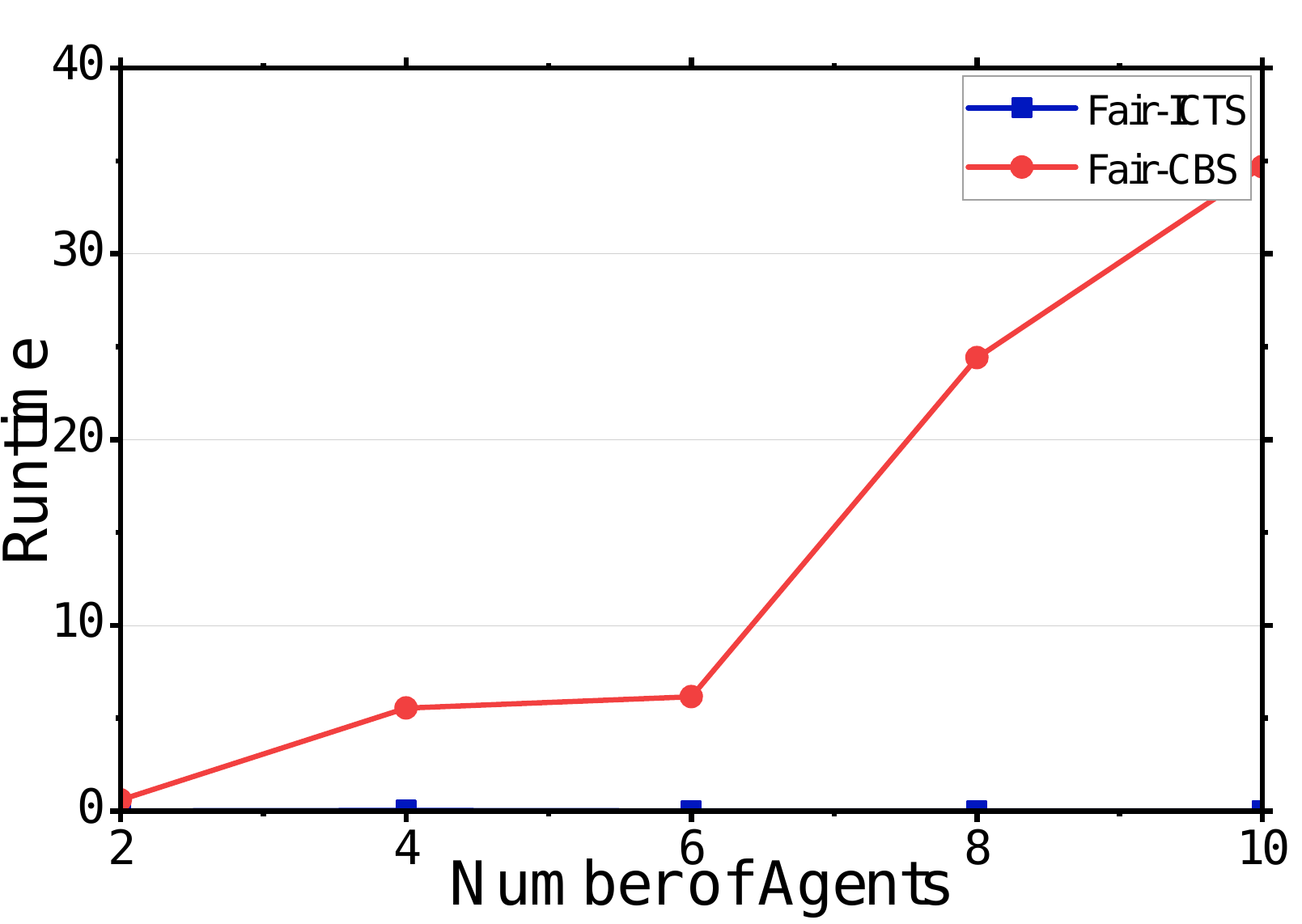} \\
\tiny{(j)} & \tiny{(k)} & \tiny{(l)} \\[4pt]

\end{tabular}

\caption{
\scriptsize
Experimental results for the proposed algorithms on four benchmark MAPF environments.  
(a)--(e) correspond to the \texttt{random32-32-20} map, 
(f)--(i) to \texttt{empty48-48}, 
(j) to \texttt{den312d}, 
% (m)--(o) to \texttt{den520d}, 
% (p)--(q) to \texttt{empty4-4}, 
and (k)--(l) to \texttt{empty16-16}.  
Specifically: 
(a) proportional fairness (success fraction), 
(b) proportional fairness (runtime), 
(c) $\epsilon$-envy fairness (runtime), 
(d) max--min fairness (success fraction), 
(e) max--min fairness (runtime) on \texttt{random32-32-20}; 
(f) proportional fairness (success fraction), 
(g) proportional fairness (runtime), 
(h) max--min fairness (success fraction), 
(i) max--min fairness (runtime) on \texttt{empty48-48}; 
(j) proportional fairness (runtime) on \texttt{den312d}; 
(k) max--min fairness (success fraction), 
and (l) max--min fairness (runtime) on \texttt{empty16-16}.  
}

\label{fig:exp_results}
% \vspace{0.05in}
\end{figure*}

\paragraph{\textbf{Number of Agents vs. Success Rate.}} Figures \ref{fig:exp_results} (a), (d), (f), (h), (k) show the variation in the success fraction with an increasing number of agents for different maps. Across all environments, the success fraction decreases as the number of agents grows, indicating a higher probability of conflicts and path congestion in denser settings. Interestingly, the success fraction is generally higher for larger maps than for smaller ones. This trend arises because larger environments provide greater spatial flexibility. Conversely, in smaller grids, the limited free space and higher agent density lead to more frequent collisions. 

% Nonetheless, in all cases, the success rate gradually decreases as the number of agents increases, reflecting the increasing complexity of coordination in denser multi-agent settings. Additionally, the success fraction is consistently higher for the Fair-CBS algorithm compared to the Fair-ICTS algorithm, indicating that Fair-CBS is more effective in finding conflict-free and fair solutions.

\paragraph{\textbf{Efficiency Test.}} Figures \ref{fig:exp_results}(b), (e), (g), (i), (j), and ($\ell$) show runtime performance as the number of agents increases. Fair-ICTS consistently exhibits low runtime across all maps, remaining below $0.04$ s on small maps, under $1.2$s on medium maps, and within $0.04$-$3$s on dense maps. In contrast, Fair-CBS incurs significantly higher runtime, ranging from $0.6$-$34.6$s on \textit{empty16-16}, up to $32$s on \textit{random-32-32-20}, $12$s on \textit{empty48-48}, and $7$-$92$s on \textit{den312d}. Fair-ICTS achieves faster runtime due to its compact search structure, Fair-CBS incurs higher computational overhead because of its explicit constraint expansion. 
% There is a trade-off: Fair-CBS achieves higher success fractions at the cost of increased runtime, whereas Fair-ICTS offers scalability and responsiveness for time-constrained applications.

\paragraph{\textbf{Scalability Test.}}
The overall trends indicate that both Fair-ICTS and Fair-CBS scale reliably up to moderate agent densities. On larger maps, the algorithms continue to maintain stable run times despite the increased combinatorial complexity, demonstrating robustness to growing spatial dimensions. However, as the number of agents becomes very large, the computational overhead would increase substantially because the search must handle a significantly expanded set of potential conflicts while simultaneously enforcing fairness constraints. 

\paragraph{\textbf{Varying $\epsilon$.}}
Figure \ref{fig:exp_results}(c) presents results for varying fairness tolerance $\epsilon$. A smaller $\epsilon$ imposes stricter fairness conditions, as $\epsilon$ increases, the algorithm finds feasible solutions more easily, resulting in higher success rates. The runtime trends further highlight this behavior. For Fair-ICTS, the runtime remains relatively stable below $1.0$s. Fair-CBS exhibits higher sensitivity reflecting the higher branching factor and constraint-handling cost. 

% \paragraph{\textbf{Additional Discussions.}}
% While the results highlight the overall effectiveness of Fair-ICTS and Fair-CBS, several observations merit further discussion. First, the scalability behaviors of the two algorithms indicate a trade-off between search structure and fairness handling: Fair-ICTS benefits from its bounded search, resulting in stable runtimes, whereas Fair-CBS exhibits larger runtime variance due to dynamic expansion of conflict and fairness constraints at higher nodes in the search tree. Second, the experiments reveal that fairness constraints influence not only the feasibility of solutions but also the structural difficulty of the search space; stricter fairness often eliminates large portions of candidate solutions, causing more backtracking in Fair-CBS but only moderately affecting Fair-ICTS. Finally, the notably lower success fractions on very large maps suggest that integrating heuristics, pruning rules or approximate fairness checks could further improve scalability for real-world deployments.

\section{Concluding Remarks}\label{sec:conclusion}
In this paper, we study fairness-driven MAPF problem, where the goal is to find non-conflicting and fair paths for all agents that maximize social welfare. We have considered three fairness notions, namely $\epsilon$-envy freeness, max-min fairness, and proportional fairness. We have proposed solutions for our problem, considering the agents are non-rational, and also developed a dominant strategy incentive compatible mechanism for rational agents. Our future study on this problem will concentrate on considering different practical constraints, such as online setup where the agents can join and leave at any point in time, considering uncertain environments, and so on. 

% \section*{\uppercase{Acknowledgements}}
\bibliographystyle{apalike}
{\small
\bibliography{example}}

@article{friedrich2024scalable,
  title={Scalable mechanism design for multi-agent path finding},
  author={Friedrich, Paul and Zhang, Yulun and Curry, Michael and Dierks, Ludwig and McAleer, Stephen and Li, Jiaoyang and Sandholm, Tuomas and Seuken, Sven},
  journal={arXiv preprint arXiv:2401.17044},
  year={2024}
}

@inproceedings{stern2019multi,
  title={Multi-agent pathfinding: Definitions, variants, and benchmarks},
  author={Stern, Roni and Sturtevant, Nathan and Felner, Ariel and Koenig, Sven and Ma, Hang and Walker, Thayne and Li, Jiaoyang and Atzmon, Dor and Cohen, Liron and Kumar, TK and others},
  booktitle={Proceedings of the International Symposium on Combinatorial Search},
  volume={10},
  number={1},
  pages={151--158},
  year={2019}
}

@article{ma2017lifelong,
  title={Lifelong multi-agent path finding for online pickup and delivery tasks},
  author={Ma, Hang and Li, Jiaoyang and Kumar, TK and Koenig, Sven},
  journal={arXiv preprint arXiv:1705.10868},
  year={2017}
}

@article{ho2019multi,
  title={Multi-agent path finding for UAV traffic management: Robotics track},
  author={Ho, Florence and Goncalves, Artur and Salta, Ana and Cavazza, Marc and Geraldes, Ruben and Prendinger, Helmut},
  year={2019},
  publisher={Association for Computing Machinery (ACM)}
}

@article{liu2024multi,
  title={Multi-agent target assignment and path finding for intelligent warehouse: A cooperative multi-agent deep reinforcement learning perspective},
  author={Liu, Qi and Gao, Jianqi and Zhu, Dongjie and Qiao, Zhongjian and Chen, Pengbin and Guo, Jingxiang and Li, Yanjie},
  journal={arXiv preprint arXiv:2408.13750},
  year={2024}
}

@inproceedings{vsvancara2019online,
  title={Online multi-agent pathfinding},
  author={{\v{S}}vancara, Ji{\v{r}}{\'\i} and Vlk, Marek and Stern, Roni and Atzmon, Dor and Bart{\'a}k, Roman},
  booktitle={Proceedings of the AAAI conference on artificial intelligence},
  volume={33},
  number={01},
  pages={7732--7739},
  year={2019}
}

@article{shahar2021safe,
  title={Safe multi-agent pathfinding with time uncertainty},
  author={Shahar, Tomer and Shekhar, Shashank and Atzmon, Dor and Saffidine, Abdallah and Juba, Brendan and Stern, Roni},
  journal={Journal of Artificial Intelligence Research},
  volume={70},
  pages={923--954},
  year={2021}
}

@inproceedings{ma2017multi,
  title={Multi-agent path finding with delay probabilities},
  author={Ma, Hang and Kumar, TK Satish and Koenig, Sven},
  booktitle={Proceedings of the AAAI Conference on Artificial Intelligence},
  volume={31},
  number={1},
  year={2017}
}

@article{sartoretti2019primal,
  title={Primal: Pathfinding via reinforcement and imitation multi-agent learning},
  author={Sartoretti, Guillaume and Kerr, Justin and Shi, Yunfei and Wagner, Glenn and Kumar, TK Satish and Koenig, Sven and Choset, Howie},
  journal={IEEE Robotics and Automation Letters},
  volume={4},
  number={3},
  pages={2378--2385},
  year={2019},
  publisher={IEEE}
}

@ARTICLE{astar,
  author={Hart, Peter E. and Nilsson, Nils J. and Raphael, Bertram},
  journal={IEEE Transactions on Systems Science and Cybernetics}, 
  title={A Formal Basis for the Heuristic Determination of Minimum Cost Paths}, 
  year={1968},
  volume={4},
  number={2},
  pages={100-107},
  doi={10.1109/TSSC.1968.300136}
}

@inproceedings{okumura2023lacam,
  title={Lacam: Search-based algorithm for quick multi-agent pathfinding},
  author={Okumura, Keisuke},
  booktitle={Proceedings of the AAAI Conference on Artificial Intelligence},
  volume={37},
  number={10},
  pages={11655--11662},
  year={2023}
}
% \section*{\uppercase{Appendix}}
\end{document}